\documentclass{tlp}
\usepackage{times}
\usepackage[dvips]{graphicx}
\usepackage{amsfonts}
\usepackage{amsmath}
\usepackage{amssymb}
\usepackage{bussproofs}
\usepackage{stmaryrd}
\usepackage{url}
\usepackage{xspace}

\newcommand{\eat}[1]{}

\newcommand{\sort}[1]{\text{\sc #1}}

\newcommand{\vect}[1]{\bar{#1}}

\newcommand{\set}[1]{\ensuremath{\mathbb{#1}}}

\newcommand{\Poss}{\ensuremath{\text{\em Poss}}}
\newcommand{\Holds}{\ensuremath{\text{\em Holds}}}

\newcommand{\Do}{\ensuremath{\text{\em Do}}}

\newcommand{\DomAx}{\ensuremath{{\cal D}}}
\newcommand{\DomAxs}{\ensuremath{{\cal D}}\ }

\newcommand{\Dom}{\ensuremath{\DomAx_{\text{dc}}}}
\newcommand{\Pre}{\ensuremath{\DomAx_{\text{Poss}}}}
\newcommand{\Init}{\ensuremath{\DomAx_{\text{Init}}}}

\newcommand{\Aux}{\ensuremath{\DomAx_{\text{aux}}}}
\newcommand{\Auxs}{\ensuremath{\DomAx_{\text{aux}}}\ }
\newcommand{\EffAx}{\ensuremath{\DomAx_{\text{Effects}}}}
\newcommand{\Eff}{\ensuremath{\Upsilon}}

\newcommand{\Program}{{\Pi}}
\newcommand{\Programs}{\ensuremath{\Pi}\ }
\newcommand{\Query}{{\Gamma}} 
\newcommand{\Querys}{\ensuremath{\Gamma}\ } 

\newcommand{\imp}{\supset}
\newcommand{\lpmi}{\subset}
\renewcommand{\iff}{\equiv}
\newcommand{\True}{\mbox{\em true\/}}

\newcommand{\At}{\mbox{\em At\/}}
\newcommand{\Agent}{\mbox{\em Agent\/}}
\newcommand{\Gold}{\mbox{\em Gold\/}}
\newcommand{\Go}{\mbox{\em Go\/}}
\newcommand{\Explore}{\mbox{\em Explore\/}}
\newcommand{\Select}{\mbox{\em Select\/}}

\newcommand{\Rmc}{\ensuremath{\mathcal{R}\xspace}}

\def\bbox{\vrule width 5 true pt height 5 true pt depth 0 pt}

\newtheorem{definition}{Definition}
\newtheorem{tdef}{Definition}
\newtheorem{ttheo}[tdef]{Theorem}
\newtheorem{tprop}[tdef]{Proposition}
\newtheorem{tlem}[tdef]{Lemma}
\newtheorem{texam}{Example}
\newtheorem{texamc}{Example}

\newenvironment{proposition}{\begin{tprop}\ \it}{\end{tprop}}

\newenvironment{example}{\begin{texam}\ \em}{\ \hfill\bbox\end{texam}}

               {\end{list}}

\newenvironment{mylistd}{\begin{list}{$\bullet$}{%
               \itemsep 0.2em%
               \labelsep 0.2em%
               \topsep   0.2em%
               \itemindent -0.8em%
               \leftmargin 0.8em}}%
               {\end{list}}

\newenvironment{mylisti}{\begin{list}{--}{%
               \itemsep 0.2em%
               \topsep   0.2em%
               \labelsep  0.4em%
               \leftmargin 1.1em}}%
               {\end{list}}

\newcommand\continuedexample{%
  \def\thetexam{1\ \contname}%
}
\newcommand\newexample{%
  \def\thetexam{2\ (Disjunctions and Substitutions in ALPprolog)}%
}
\newcommand\newproposition{%
  \def\thetdef{1\ (Soundness of ALPprolog)}%
}

\title{ALPprolog --- A New Logic Programming Method for Dynamic Domains}

\author[Conrad Drescher and Michael Thielscher]{CONRAD DRESCHER \\
  Computing Laboratory,  University of Oxford, UK\\
  \email{Conrad.Drescher@comlab.ox.ac.uk}
  \and
  MICHAEL THIELSCHER\\
  School of Computer Science and Engineering, The University of New South Wales, Australia\\
  \email{mit@cse.unsw.edu.au}
}

\begin{document}

\maketitle

\noindent
{\bf Note:} This article has been published in \emph{Theory and Practice of Logic Programming},volume 11, issue 4-5, pp. 451-468, \copyright Cambridge University Press.
\enlargethispage{6ex}


\begin{abstract}
Logic programming is a powerful paradigm for programming autonomous agents in dynamic domains,
as witnessed by languages such as Golog and Flux.
In this work we present ALPprolog, an expressive, yet efficient, 
logic programming language for the online control of agents 
that have to reason about incomplete information and sensing actions.
\end{abstract}

\section{Introduction}

Programming autonomous agents that behave intelligently is one of the key challenges of Artificial Intelligence.
Because of its declarative nature, and high level of abstraction, logic programming is a natural choice for this task.
This is witnessed by e.g.\ the two major exponents of agent programming languages that are based on classical logic 
programming, namely Golog \cite{levesq:golog} and Flux \cite{A:16}.

Both these languages combine a language for specifying the agent's behaviour with an axiomatic theory that describes 
the agent's environment.
In the case of Golog the strategy language is procedural in nature (though implemented in Prolog), and the action 
theory is the classical Situation Calculus \cite{MH69} in
Reiter's version \cite{Reiter01}.
For Flux the strategy language is full classical logic programming, and the action theory is the more recent 
Fluent Calculus \cite{A:11}.

In a recent work \cite{alp} we have developed Agent Logic Programs (ALPs), a new declarative strategy language
that is based upon a proof calculus in the style of classical SLD-resolution.
Contrary to Golog and Flux the ALP framework is parametric in the action theory\/:
any background theory that allows to infer when an action is applicable, and what the effects of the action are, can be used.
Exploiting this generality we have recently \cite{thielscher:KR10a} been able to give a semantics for the 
BDI-style language AgentSpeak \cite{bordin:progra}.
Another distinctive feature of the theoretical framework is the elegant handling of incomplete information for offline planning
via {\em disjunctive substitutions}.
By default, ALPs are combined with our new Unifying Action Calculus (UAC) \cite{thielscher:AIJ11}
that encompasses the major logical action calculi, including both the Situation Calculus and the Fluent Calculus, as well as many planning domain description languages. The ALP formalism stays entirely within classical logic.

The implementation of any fragment of the ALPprolog framework consists of (1) an implementation of the proof calculus,
and (2) an action theory reasoner. Existing mature Prolog technology can be used out of the box for (1) unless disjunctive substitutions enter the picture.
For (2) we can also exploit existing technology\/: E.g.\ Golog implements a fragment of the Situation Calculus, and Flux handles a fragment of the Fluent Calculus.
In \cite{update} the implementation of a Description Logic-based fragment of the Fluent Calculus is described.

In this work we present ALPprolog, 
where the underlying action theory is an essentially propositional version of the Fluent Calculus in the 
UAC that includes a simple, yet powerful model of sensing.
ALPprolog is intended for the online control of agents, where actions are immediately executed.
This starkly contrasts with offline reasoning, where agents may make assumptions to see where these are leading.
ALPprolog was developed specifically for the efficient handling of large ground state representations,
something that we consider to be practically useful.
To this end ALPprolog combines strong-points of Golog and Flux\/:

\begin{itemize}
\item From Golog it takes the representation of the agent's state knowledge in full propositional logic 
  via prime implicates; and
\item From Flux it takes the principle of progression\/: The agent's state knowledge is updated upon the 
  execution of an action.
  In standard Golog the agent's initial state knowledge is never 
  updated.\footnote{But there is a version of Golog where the initial state is periodically updated \cite{sardin:wumpus}.}
  Instead, queries referring to later time-points are rewritten until they can be evaluated against the 
  initial state knowledge,
  something which becomes a hindrance to good performance as the sequence of executed actions grows.
\end{itemize}

We emphasise that ALPprolog is an agent programming language in the spirit of classical logic programming in Prolog\/:
The straightforward operational semantics provides the programmer with a powerful means to actively determine the 
sequence of actions that an agent executes.
ALPprolog\footnote{The name is a play on ALPs, propositional logic, and the implementation in plain Prolog.}
can be obtained at \,\url{alpprolog.sourceforge.net}.

The remainder of this paper is organised as follows\/:
In Section~\ref{sec:alp} we recall the basics of the ALP framework, and
in Section~\ref{sec:alpprolog} we introduce ALPprolog.
We evaluate the performance of ALPprolog in Section~\ref{sec:eval}, and 
conclude in Section~\ref{sec:conc}.

\section{ALPs in a Nutshell}
\label{sec:alp}

The purpose of agent logic programs is to provide high-level control
programs for agents using a combination of declarative programming
with reasoning about actions. The syntax of these programs is
kept very simple\/: standard (definite) logic programs (see e.g.\ \cite{Lloy87}) are augmented
with just two
special predicates, one --- written \,\verb:do(:$\!\alpha\!$\verb:): --- to
denote the execution of an action by the agent, and one --- written
\,\verb:?(:$\!\varphi\!$\verb:): --- to verify properties against (the
agent's model of)
the state of its environment. This model, and how it is affected by
actions, is defined in a separate action theory. This allows for a
clear separation between the agent's strategic behaviour (given by the
agent logic program itself) and the underlying theory about the
agent's actions and their effects. Prior to giving the formal
definition, let us illustrate the idea by an example agent
logic program.

\newcounter{EX}
\setcounter{EX}{\value{texam}}
\begin{example}
\label{ex:alp}
Consider an agent whose task is to find gold in a maze. For the sake
of simplicity, the states of
the environment shall be described by a single {\em fluent\/} (i.e.,
state property)\/: $\At(u,x)$ to
denote that~$u\in\{\Agent,\Gold\}$ is at location~$x$. The agent can
perform the action $\Go(y)$ of going to location~$y$, which is
possible if~$y$ is adjacent to
the current
location of the agent. 
The following ALP describes a simple search strategy 
via a given list of locations (choice points) that
the agent may visit, and an ordered collection of backtracking
points. We follow the Prolog convention of writing variables with a leading uppercase letter.

{\small
\begin{verbatim}
 explore(Choicepoints,Backtrack) :-         % finished, if
    ?(at(agent,X)), ?(at(gold,X)).          % gold is found

 explore(Choicepoints,Backtrack) :-           
    ?(at(agent,X)),
    select(Y,Choicepoints,NewChoicepoints), % choose a direction
    do(go(Y)),                              % go in this direction
    explore(NewChoicepoints,[X|Backtrack]). % store the choice

 explore(Choicepoints,[X|Backtrack]) :-     % go back one step
    do(go(X)),
    explore(Choicepoints,Backtrack).

 select(X,[X|Xs],Xs).
 select(X,[Y|Xs],[Y|Ys]) :- select(X,Xs,Ys).
\end{verbatim}}
\noindent
Suppose we are given 
a
list of choice points \,\verb:C:\,,
then the query \,\verb#:- explore(C,[])#\, lets the agent
systematically search for gold from its current location\/: the first
clause describes the base case where the agent is successful; the
second clause lets the agent select a new location from the list of
choice points and go to this location (the declarative semantics and
proof theory for \,\verb:do(:$\!\alpha\!$\verb:):\, will require that the
action is possible at the time of execution); and the
third clause sends the agent back using the latest backtracking
point.
\end{example}

The example illustrates two distinct features of ALPs\/:
(1) The agent strategy is defined by
a logic program that may use arbitrary function and predicate symbols
in addition to the signature of the underlying action theory.
(2) The update of
the agent's belief according to the effects of its actions is not part
of the strategy. 
Formally, ALPs are defined as follows.

\begin{definition}
Consider an action theory signature $\Sigma$ that includes the pre-defined sorts
\sort{action} and \sort{fluent}, along with a logic program
signature~$\Pi\!$.
\begin{mylistd}
\item {\em Terms\/} are from $\Sigma\cup\Pi\!$.
\item If \,\verb:p:\, is an $n\!$-ary relation symbol from $\Pi$ and
  \,{\tt t}$\!_1,...,\!${\tt t}$\!_n\!$ are terms, then 
  \,{\tt p(t}$\!_1,...,\!${\tt t}$\!_n)\!$ is a {\em program atom\/}.
\item \verb:do(:$\!\alpha\!$\verb:):\, is a {\em program atom\/} if
  $\alpha$ is an \sort{action} term in~$\Sigma\!$.
\item \verb:?(:$\!\varphi\!$\verb:):\, is a {\em program atom\/} if
  $\varphi$ is a {\em state property\/} in~$\Sigma\!$, that is, a
  formula (represented as a term) based on the \sort{fluent}s
  in~$\Sigma\!$.
\item Clauses, programs, and queries are then defined as usual for
  definite logic programs, with the restriction that the two special
  atoms cannot occur in the head of a clause.\ \hfill\bbox
\end{mylistd}
\end{definition}


\subsection{Declarative Semantics\/: Program + Action Theory}

The semantics of an ALP is given in two steps. First,
the program needs to be ``temporalised,'' making explicit
the state change that is implicit in the
use of the two special predicates,
\,\verb:do(:$\!\alpha\!$\verb:):\, and
\,\verb:?(:$\!\varphi\!$\verb:):.
Second, the resulting program is combined with an action
theory as the basis for evaluating these two special predicates.
The semantics is then the classical logical semantics of the expanded program together with the action theory.

Time is incorporated into a program through macro-expansion\/:
two arguments of sort \sort{time}\footnote{Which specific concept of
  time is being
  used depends on how the sort \sort{time} is defined in the
  underlying action theory, which may be
  branching  (as, e.g., in the Situation Calculus) or
  linear (as, e.g., in the Event Calculus).}
are added to every regular
program atom~$p(\vect x)\!$, and then $p(\vect x,s_1,s_2)$ is understood as
restricting the truth of the atom to the temporal interval between
(and including) $s_1$ and $s_2\!$.
The two special atoms receive special
treatment\/: \,\verb:?(:$\!\varphi\!$\verb:):\, is re-written to
$\Holds(\varphi,s)\!$, with the intended meaning that~$\varphi$
is true at~$s\!$; and \,\verb:do(:$\!\alpha\!$\verb:):\,
is mapped onto $\Poss(\alpha,s_1,s_2)\!$, meaning that action~$\alpha$ can
be executed at~$s_1$ and that its execution ends in~$s_2\!$. 
The formal definition is as follows.

\begin{definition} \label{d:expansion}
For a clause {\tt H\,:-\,B$\!_1\!$,...,B$\!_n\!$}
($\!n\geq0\!$), let $s_1, \ldots, s_{n+1}$ be variables of sort
\sort{time}.
\begin{mylistd}
\item For $i=1,\ldots,n\!$, if \,{\tt B}$\!_i$ is of the form
  \begin{mylisti}
  \item {\tt p(t$\!_1\!$,...,t$\!_m\!$)}, expand to $P(t_1,\ldots,t_m,
  s_i, s_{i+1})$.
  \item {\tt do(}$\!\alpha\!${\tt )}, expand to $\Poss(\alpha, s_i, s_{i+1})$.
  \item {\tt ?(}$\!\varphi\!${\tt )}, expand to $\Holds(\varphi,s_i) \land s_{i+1} = s_i$.
  \end{mylisti}
\item The head atom \,{\tt H\,=\,p(t$\!_1\!$,...,t$\!_m\!$)}\, is expanded
  to $P(t_1,\ldots, t_m, s_1, s_{n+1})\!$.
\item The resulting clauses are understood as universally quantified
  implications.
\end{mylistd}
\noindent
Queries are expanded exactly like clause bodies, except that 
\begin{mylistd}
\item a special constant $S_0$ --- denoting the earliest time-point in the
  underlying action theory --- takes the place of $s_1$;
\item the resulting conjunction is existentially quantified.\ \hfill\bbox
\end{mylistd}
\end{definition}

\continuedexample
\begin{example}
The example program of the preceding section is
understood as the following axioms, which for notational convenience we have simplified in that all equations between {\sc time} variables have been
applied and then omitted.
\[ \begin{array}{rcl}
  (\forall)\Explore(c,b,s_1,s_1)& \lpmi & \Holds(\At(\Agent,x),s) \wedge \Holds(\At(\Gold,x),s_1) \\
  (\forall)\Explore(c,b,s_1,s_4) & \lpmi & \Holds(\At(\Agent,x),s_1) \wedge \Select(y,c,c',s_1,s_2)\wedge \\
  && \Poss(\Go(y),s_2,s_3)\wedge \Explore(c',[x|b],s_3,s_4) \\
  (\forall)\Explore(c,[x|b],s_1,s_3) &\lpmi & \Poss(\Go(x),s_1,s_2)\wedge \Explore(c,b,s_2,s_3) \\ 
  \multicolumn{3}{l}{(\forall)\Select(x,[x|x'],x',s_1,s_1)\lpmi\True} \\
  \multicolumn{3}{l}{(\forall)\Select(x,[y|x'],[y|y'],s_1,s_2)\lpmi\Select(x,x',y',s_1,s_2)}
\end{array} \]
The resulting theory constitutes a purely logical axiomatisation of the
agent's strategy, which provides the basis for logical entailment.
For instance, macro-expanding the
query \,\verb#:- explore(C,[])# from the above example results in the temporalised logical formula 
$(\exists s)\,\Explore(C,[\,],S_0,s)$.
If this formula follows from
the axioms above, then that means that the strategy can be
successfully executed, starting at~$S_0$, for the given list of
choice points~$C$. Whether this is actually the case of course
depends on the additional action theory that is needed to evaluate
the special atoms $\Holds$ and $\Poss$ in a macro-expanded program.
\end{example}

Macro-expansion provides the first part of the declarative semantics
of an agent logic program; the second part is given by an action
theory in form of a logical axiomatisation of actions and their
effects. The overall declarative semantics of agent logic programs is
given by the axiomatisation consisting of the action theory and the expanded program.

Let us next introduce the fragment of the UAC corresponding to the Fluent Calculus.
The UAC that is used to axiomatise the action theory is based on many-sorted first order logic with 
equality and the four sorts \sort{time}, 
\sort{fluent}, \sort{object}, and \sort{action}. By convention variable 
symbols $s\!$, $f\!$, $x\!$, and $a$ are used for terms of sort \sort{time}, \sort{fluent}, \sort{object}, 
and \sort{action}, respectively.
Fluents are reified, and
the standard predicate 
$\Holds:\sort{fluent}\times\sort{time}$ 
indicates whether a fluent is true at a particular time.
The predicate 
$\Poss(a,s_1,s_2)$ 
means that action~$\alpha$ can be executed at~$s_1$ and that its execution ends in~$s_2$.
The number of function symbols into sorts \sort{fluent} and \sort{action} is finite.

\begin{definition}[Action Theory Formula Types]
\label{def:uac}
We stipulate that the following formula types are used by action theories\/:
\begin{itemize}
\item State formulas express what is true at particular times\/:
  A {\em state formula $\Phi[\vect s\,]$ in~$\vect s$\/}
  is a first-order formula with free variables $\vect s$ where
  \begin{itemize}
  \item for each occurrence of $\Holds(f,s)$ we have $s\in\vect s$;
  \item predicate $\Poss$ does not occur.
  \end{itemize}
  A state formula is {\em pure} if it does not mention predicates other than $\Holds$.
\item A {\em state property} $\phi$ is an expression built from the standard logical 
  connectives and terms $F(\vect x)$ of sort \sort{fluent}.
  With a slight abuse of notation, by $\Holds(\phi,s)$ we denote the state formula obtained from state property $\phi$ by 
  replacing every occurrence of a fluent $f$ by $\Holds(f,s)$. In an expanded program $\Program$ we always treat $\Holds(\phi,s)$ as atomic.
  State properties are used by agent logic programs in \verb:?(Phi): atoms.
\item The {\em initial state axiom} is a state formula $\phi(S_0)$ in $S_0$, where $S_0$ denotes the initial situation.
\item An {\em action precondition axiom\/} is of the form
  \[(\forall)\Poss(A(\vect x),s_1,s_2)\,\equiv\,\pi_A[s_1] \land s_2 = \Do(A(\vect x),s_1),\]
  where $\pi_A[s_1]$ is a state formula in~$s_1$ with free variables among $s_1,\vect x$.
  This axiom illustrates how different actions lead to different situation terms $\Do(A(\vect x),s_1)$.
  Situations constitute the sort {\sc time} in the Fluent Calculus and provide a branching time
  structure.
\item {\em Effect axioms} are of the form
  \begin{flalign*}
      &\Poss(A(\vec x),s_1,s_2)\imp\\
      &\quad \bigvee_k (\exists \vec y_k)(\Phi_k[s_1]\wedge(\forall f)[(\bigvee_i f=f_{ki}\vee(\Holds(f,s_1)\wedge \bigwedge_j f\neq g_{kj}))\\
      &\quad \quad \quad \quad \quad \quad \quad \quad \quad \quad \quad \quad \quad \quad \quad \quad \quad \quad \quad \quad \quad \quad \quad \iff\Holds(f,s_2)]).
  \end{flalign*}
  Such an effect axiom has $k$ different cases that can apply --- these are identified by the case selection formulas 
  $\Phi_k[s_1]$ which
  are state formulas in~$s_1$ with free variables among $s_1,\vec x,\vec y_i$.
  The $f_{ki}$ (and $g_{kj}$, respectively) are fluent terms with variables among
    $\vec x,\vec y_k$ and describe the positive (or, respectively, negative) 
  effects of the action, given that case $k$ applies.
\item {\em Domain constraints} are universally quantified state formulas $(\forall s)\delta[s]$ in~$s$.
\item {\em Auxiliary axioms} are domain-dependent, but time-independent, additional axioms such as e.g.\ an axiomatisation 
  of finite domain constraints.
\end{itemize}
\end{definition}

An action theory \DomAx\ is given by an initial state axiom $\Init$, finite sets $\Pre$ and $\EffAx$ of precondition and
effect axioms. Moreover domain constraints $\Dom$ and auxiliary axioms $\Aux$ may be included.
For illustration, the following is a background
axiomatisation for our example scenario as
a basic Fluent Calculus theory in the UAC.

\continuedexample
\begin{example}
Our example program can be supported by the following domain theory.
\begin{mylistd}
  \item Initial state axiom
\[
  \Holds(\At(\Agent,1),S_0)\wedge\Holds(\At(\Gold,4),S_0)
\]
  \item Precondition axiom
\[ \begin{array}{ll}
  \Poss(\Go(y),s_1,s_2)\iff &(\exists x)(\Holds(\At(\Agent,x),s_1)\wedge (y=x+1\vee y=x-1)) \\
  &\wedge\, s_2=\Do(\Go(y),s_1)
\end{array} \]
  \item Effect axiom
    \begin{flalign*}
      &\Poss(\Go(y),s_1,s_2) \imp\\
      & \quad (\exists x) ( \Holds(\At(\Agent,x),s_1)\ \land\\
      & \quad [ (\forall f)  \Holds(f,s_2 ) \equiv ( \Holds(f,s_1) \lor f = \At(\Agent,y) ) \land f \neq \At(\Agent,x) ] ).
    \end{flalign*}
\end{mylistd}
Given this (admittedly very simple, for the sake of illustration)
specification of the background action theory, the axiomatisation of
the agent's strategy from above entails, for example,
$(\exists s)\,\Explore([2,3,4,5],[\,],S_0,s)$.
This can be shown as follows.
First, observe that the background
theory entails
\[
  \Holds(\At(\Agent,4),S)\wedge\Holds(\At(\Gold,4),S),
\]
where $S$ denotes the situation term $\Do(\Go(4),\Do(\Go(3),\Do(\Go(2),S_0)))$.
It follows that
$\Explore([5],[3,2,1],S,S)$ according to the first clause of
our example ALP. Consider, now, the situation
$S'=\Do(\Go(3),\Do(\Go(2),S_0))$, then action theory and strategy
together imply
\[ 
  \Holds(\At(\Agent,3),S')\wedge\Select(4,[4,5],[5],S',S')\wedge \Poss(\Go(4),S',S)
\]
By using this in turn, along with $\Explore([5],[3,2,1],S,S)$ from above, according to the second program clause
we obtain $\Explore([4,5],[2,1],S',S)$.
Continuing this line of reasoning, it can be shown that
\[ \begin{array}{ll}
  & \Explore([3,4,5],[1],\Do(\Go(2),S_0),S) \\
  \mbox{and hence,} & \Explore([2,3,4,5],[\,],S_0,S)
\end{array} \]
This proves the claim that
$(\exists s)\,\Explore([2,3,4,5],[\,],S_0,s)$.
On the other hand e.g.\ the query $(\exists s)\,\Explore([2,4],[\,],S_0,s)$
is {\em not\/}
entailed under the given background theory\/:
Without
location~$3$ among the choice points, the strategy does
not allow the agent to reach the only location that is known to house gold.
\end{example}

\subsection{Operational Semantics\/: Proof Calculi}

We have developed two sound and complete proof calculi for ALPs that both 
assume the existence of a suitable reasoner for the underlying action theory \cite{alp}.

The first proof calculus is plain SLD-resolution, only that \Holds- and \Poss-atoms are evaluated against the action
theory. This calculus is sound and complete if the underlying action theory has the witness property\/:
That is, whenever $\DomAx \vDash (\exists x) \phi(x)$ then there is a substitution $\theta$ such that
$\DomAx \vDash (\forall) \phi(x)\theta$. Note that in general action theories may violate the witness property,
as they may include disjunctive or purely existential information; consider e.g.\ the case
$\Holds(\At(\Gold,4),S_0) \lor \Holds(\At(\Gold,5),S_0)$, where the exact location of the gold is unknown.

Hence the second proof calculus, intended for the general case, resorts to constraint logic programming, and the notion of a
disjunctive substitution\/: Still assuming that the gold is located at one of two locations 
the query $\DomAx \vDash (\exists x)\Holds(\At(\Gold,x)$
can now be answered positively via the disjunctive substitution $x \rightarrow 4 \lor x \rightarrow 5$.
Disjunctive substitution together with the respective principle of reasoning by cases are a powerful
means for inferring conditional plans.

For the online control of agents, however, assuming a particular case
is unsafe. But if we use the plain SLD-resolution-based ALP proof calculus on top of action theories that lack the witness
property we obtain a nice characterisation of cautious behaviour in a world of unknowns (albeit at the cost of 
sacrificing logical completeness). For ALPprolog this is the setting that we use.

In both proof calculi we adopt the "leftmost" computation rule familiar from Prolog.
This has many advantages\/:
First, it simplifies the implementation, as this can be based on existing mature Prolog technology.
Second, state properties can always be evaluated against a description of the "current" state.
Last, but not least, this ensures that actions are executed in the order intended by the programmer ---
this is of no small importance for the online control of agents.

\section{ALPprolog}
\label{sec:alpprolog}

We next present ALPprolog --- 
an implementation of the ALP framework atop of action theories in a version of the Fluent Calculus that 

\begin{itemize}
\item uses (a notational variant of) propositional logic for describing state properties;
\item is restricted to actions with ground deterministic effects; and
\item includes sensing actions.
\end{itemize}

The intended application domain for ALPprolog is the online control of agents in dynamic domains with incomplete information.

\subsection{ALPprolog Programs}

An ALPprolog program is an ALP that respects the following restrictions on the \verb#?(Phi)# atoms in the program\/:

\begin{itemize}
\item All occurrences of non-fluent expressions in $\phi$ are positive.
\item So called sense fluents $S(\vec x)$ that represent the interface to a sensor 
  may only occur in the form \verb#?(s(X))#. Sense fluents are formally introduced below.
\end{itemize}

Because ALPprolog programs are meant for online execution the programmer must ensure that no 
backtracking over action executions occurs, by inserting cuts after all action occurrences. 
Observe that this applies to sensing actions, too.
It is readily checked that --- after the insertion of cuts --- the ALP from example~\ref{ex:alp} 
satisfies all of the above conditions.

\subsection{Propositional Fluent Calculus}

In this section we introduce the announced propositional fragment of the Fluent Calculus.
The discussion of sensing is deferred until section~\ref{ssec:sens}.

For ease of modelling we admit finitely many ground terms for fluents and objects,
instead of working directly with propositional letters.
An action domain \DomAxs is then made propositional by
including the respective domain closure axioms.
For actions, objects, and fluents unique name axioms are included ---
hence we can avoid equality reasoning.

The basic building block of both the propositional Fluent Calculus and ALPprolog are the so-called prime implicates 
of a state formula $\phi(s)$\/:

\begin{definition}[Prime Implicate]
A clause $\psi$ is a {\em prime implicate\/} of $\phi$ iff it
is entailed by $\phi\!$,
is not a tautology, and
is not entailed by another prime implicate.
\end{definition}

The prime implicates of a formula are free from redundancy --- all tautologies and implied clauses have been deleted.
For any state formula an equivalent {\em prime state formula} can be obtained by first transforming the state formula 
into a set of clauses, and by then closing this set 
under resolution, and the deletion of subsumed clauses and tautologies. 

Prime state formulas have the following nice property\/: 
Let $\phi$ be a prime state formula, and let $\psi$ be some clause (not mentioning auxiliary predicates); 
then $\psi$ is entailed by $\phi$ if and only if it is subsumed by some prime implicate in $\phi$,
a fact that has already been exploited for Golog \cite{Reiter01,reiter:knowle}.
This property will allow us to reduce reasoning about state knowledge in ALPprolog to simple list look-up operations.

Formally the propositional version of the Fluent Calculus is defined as follows.

\begin{definition}[Propositional Fluent Calculus Domain]
We stipulate that the following properties hold in propositional Fluent Calculus domains\/:

\begin{itemize}
\item The initial state \Init\ is specified by a ground prime state formula.
\item The state formulas $\phi(s_1)$ in action preconditions $\Poss(a,s_1,s_2) \equiv \phi(s_1) \land s_2 = \Do(a,s_1)$
  are prime state formulas.
\item The effect axioms
  are of the form
  \begin{flalign*}
      &\Poss(A(\vec x),s_1,s_2)\imp\\
      &\quad \bigvee_k (\Phi_k[s_1]\wedge(\forall f)[(\bigvee_i f=f_{ki}\vee(\Holds(f,s_1)\wedge \bigwedge_j f\neq g_{kj}))\\
      &\quad \quad \quad \quad \quad \quad \quad \quad \quad \quad \quad \quad \quad \quad \quad \quad \quad \quad \quad \quad \quad \quad \quad \iff\Holds(f,s_2)]),
    \end{flalign*}
  where each $\Phi_k[s_1]$ is a prime state formula. This implies that existentially quantified variables 
  that may occur in case selection formulas (cf.\ definition~\ref{def:uac})
  have been eliminated by introducing additional cases.
\item Only so-called {\em modular} domain constraints \cite{herzig:domain} may be included.
  Very roughly, domain constraints are modular if they can be compiled into the agent's initial state knowledge,
  and the effect axioms ensure that updated states also respect the domain constraints.
  In the Fluent Calculus this holds if the following two conditions are met \cite{thielscher:AIJ11}\/:
  Condition~(\ref{e:i1}),
  says that for a state that is consistent with the
  domain constraints and in which an action~$A(\vec x)$ is applicable, the
  condition~$\Phi_i[S]$ for at least one case~$i$ in the effect axiom for~$A$ holds. 
  Condition (\ref{e:i2}) requires that any possible
  update leads to a state that satisfies the domain constraints.
  Formally, let $S,T$ be constants of sort~\sort{Time}.
  $\Dom$ the domain constraints,
  $\Pre$ the precondition axioms, and $\EffAx$ the effect axioms.
  The following must hold for every action $A(\vec x)$\/:
  There exists $i=1,\ldots,n$ such that
  \begin{equation} \label{e:i1}
    \models\Dom[S]\wedge\pi_A[S]\wedge(\exists\vec y_i)\Phi_i[S],
  \end{equation}
  and for every such $i$,
  \begin{equation} \label{e:i2}
    \models\Dom[S]\wedge\pi_A[S]\wedge\Eff_i[S,T]\imp\Dom[T].
  \end{equation}
  Non-modular, fully general domain constraints greatly complicate reasoning.
\item Auxiliary time-independent axioms may be included if they can faithfully be represented in the Prolog dialect 
  underlying the implementation.
  This deliberately sloppy condition is intended to allow the programmer to use her favourite Prolog library.
  However, we stipulate that auxiliary predicates occur only positively outside of \Auxs in the action domain \DomAx in
  order to ensure that they can safely be evaluated by Prolog.
  They also must not occur in the initial state formula at all.
  The update mechanism underlying ALPprolog can handle only ground effects.
  Hence, if auxiliary atoms are used in action preconditions, case selection formulas of effect axioms, 
  then it is the burden of the programmer to ensure that these predicates always evaluate to ground terms
  on those variables that also occur in the action's effects.
\end{itemize}
\end{definition}

On the one hand clearly every propositional Fluent Calculus domain can be transformed to this form.
On the other hand it is well known that in general
compiling away the quantifiers in a state formula can result in an exponential blow-up,
as can the conversion to conjunctive normal form.
We believe that the simplicity of reasoning with prime implicates outweighs this drawback.

Propositional action domains can still be non-deterministic.
For example, for an applicable action 
two different cases may be applicable at the same time.
The resulting state would then be determined only by the disjunction of the cases' effects.
What is more, it would be logically unsound to consider only the effects of one of the cases.
For the online control of agents in ALPprolog we stipulate that for an applicable action at most a single case applies,
greatly simplifying the update of the agent's state knowledge.

\begin{definition}[Deterministic Propositional Fluent Calculus]
\label{def:fcp}
A propositional Fluent Calculus domain is deterministic if the following holds\/:
Let $a$ be an applicable ground action.
Then there is at most one case of the action that is applicable in the given state.
\end{definition}

For example, an action theory is deterministic if for each effect axiom 
all the cases are mutually exclusive. Next assume we have an applicable deterministic action with e.g.\ two case selection 
formulas $\phi(s)$ and $\neg \phi(s)$,
where neither case is implied by the current state. Here, instead of updating the current state with the disjunction
of the respective effects, ALPprolog will employ incomplete reasoning.

\subsection{Propositional Fluent Calculus with Sensing}
\label{ssec:sens}

We make the following assumptions concerning sensing\/:
At any one time, a sensor may only return a single value from a fixed set $\mathcal{R}$ of ground terms, 
the {\em sensing results}. However, the meaning of such a sensing result may depend upon the concrete situation of
the agent. 

\continuedexample
\begin{example}
Assume that now one of the cells in the maze contains a deadly threat to our gold-hunting agent.
If the agent is next to a cell containing the threat she perceives a certain smell,
otherwise she doesn't\/: 
She can sense whether one of the neighbouring cells is unsafe;
but the actual neighbouring cells are only determined by the agent's current location.
\end{example}


\begin{definition}[Sensor Axiom]
\label{def:sens}
A sense fluent $S(x)$ is a unary fluent that serves as interface to the sensor. We assume the sort \sort{sensefluent} to be a subsort of sort \sort{fluent}.
A sensor axiom then is of the form
\[
(\forall s, x,\vec y) \Holds(S(x),s) \equiv \bigvee_{R \in \mathcal{R}} x = R \land \phi(x,\vec y,s) \land \psi(x, \vec y, s),
\]
for a ground set of sensing results $\mathcal{R}$.
Here $\phi(x, \vec y, s)$ is a prime state formula that selects a meaning of the sensing result $R$,
whereas the pure prime state formula $\psi(x, \vec y, s)$ describes the selected meaning.
We stipulate that sensor axioms (which are a form of domain constraint) may only be included if they are modular.
\end{definition}

Clearly $\phi(x, \vec y, s)$ should be chosen so as to be uniquely determined in each state.
If auxiliary axioms are used in $\phi(x, \vec y, s)$ then again the programmer must ensure that these
evaluate to ground terms in order that a ground state representation can be maintained.

\continuedexample
\begin{example}
The following is the sensor axiom for our gold-hunter\/:
\begin{align*}
(\forall) &\Holds(\text{\em PerceiveSmell}(x),s) \equiv \\
& x = \text{\em true}\land \Holds(\text{\em At}(\Agent,y),s) \land \text{\em Neighbours}(y,\vec z) \land \!\bigvee_{z\in\vec z}\!\Holds(\text{\em ThreatAt}(z),s)\\
& \lor \\
& x = \text{\em false}\land \Holds(\text{\em At}(\Agent,y),s) \land \text{\em Neighbours}(y,\vec z) \land \!\bigwedge_{z\in\vec z} \!\neg \Holds(\text{\em ThreatAt}(z),s))
\end{align*}
\end{example}

Theoretically, the combination of sensing with the online control of an agent is quite challenging\/:
It is logically sound to to consider the disjunction of all possible sensing results for offline reasoning.
In the online setting, however, upon the observation of a sensing result we henceforth have to accept this result as 
being true;
that is, at runtime we {\em add} the result to the action theory, something which is logically unsound.
On the other hand, it also does not make sense to stipulate that the sensing result be known beforehand.

\subsection{Action Theory Representation}

We continue by describing how the underlying action theory is represented in ALPprolog.
As basic building block we need a representation for prime state formulas.
For notational convenience we will represent $(\neg)\Holds(f,s)$ literals by the (possibly negated) fluent terms only,
and, by an abuse of terminology, we will call such a term $(\neg) f$ a fluent literal.
A convenient Prolog representation for such a state 
formula is a list, where
each element is either a literal (i.e.\ a unit clause) or 
a list of at least two literals (a non-unit clause).
In the following we call such a list a PI-list. 

\begin{definition}[Action Theory Representation]
Action theories as defined in definition~\ref{def:fcp} are represented in ALPprolog as follows\/:
\begin{itemize}
\item The initial state is specified by a Prolog fact \verb#initial_state(PI-List).#, 
  where \verb#PI-List# mentions only ground fluent literals. Domain constraints other than sensor axioms
  have to be compiled into \verb#PI-List#.
\item a Prolog fact \,\verb#action(A,Precond,EffAx).#, for each action $a$, has to be included, where
  \begin{itemize}
  \item \verb#A# is an action function symbol, possibly with object terms as arguments;
  \item \verb#Precond# is a PI-list, the action's precondition;
  \item \verb#EffAx# is a list of cases for the action's effects with each case being a pair \verb#Cond-Eff#, where
    the effect's condition \verb#Cond# is a PI-list, and the effects \verb#Eff# are a list of 
    fluent literals; and
  \item all variables in \verb#EffAx# also occur in \verb#Precond#.
  \end{itemize}

\item If present, auxiliary axioms \Auxs are represented by a set of Prolog clauses. 
  The predicates defined in the auxiliary axioms must be declared explicitly by a fact \verb#aux(Aux).#, 
  where \verb#Aux# denotes the listing of the respective predicate symbols.
\end{itemize}

The sensor axioms are represented as Prolog facts \verb#sensor_axiom(s(X),Vals).#, where

\begin{itemize}
\item \verb#s# is a sense fluent with object argument \verb#X#; and
\item \verb#Vals# is a list of \verb#Val-Index-Meaning# triples, where
\begin{itemize}
\item \verb#Val# is a pair \verb#X-result_i#, where \verb#result_i# is the observed sensing result;
\item \verb#Index# is a PI-list consisting of unit clauses; and
\item \verb#Meaning# is a PI-list, mentioning only fluent literals and only variables from \verb#Val# and \verb#Index#.
\end{itemize}
\end{itemize}

The sense fluents have to be declared explicitly by a fact \verb#sensors(Sensors).#, 
where \verb#Sensors# is a listing of the respective function symbols.
This is necessary in order to distinguish sense fluents, ordinary fluents, and auxiliary predicates in PI-lists.
\end{definition}

\subsection{Reasoning for ALPprolog}

Reasoning in ALPprolog works as follows\/:
For evaluating the program atoms we readily resort to Prolog.
The reasoner for the action theory 
is based on the principle of progression.
Setting out from the initial state,
upon each successful evaluation of an action's precondition against the current state description,
we update the current state description by the action's effects.

Reasoning about the action comes in the following forms\/:

\begin{itemize}
\item Given a ground applicable action $a$, from the current state description $\phi(s_1)$ and the 
  action's positive and negative effects
  compute the description of the next state $\psi(s_2)$ (the update problem).
\item Given a description $\phi(s)$ of the current state, 
  check whether $\{\phi(s)\} \cup \Aux \vDash \psi(s)$, where $\psi(s)$ is some state formula in $s$,
  but not a sense fluent (the entailment problem).
\item For a sensing action, i.e.\ a query $\Holds(S(x),s)$, integrate the sensing results observed
  into the agent's state knowledge (the sensing problem).
\end{itemize}

In the following we consider each of these
reasoning problems in turn.

\subsubsection{The Update Problem}

It turns out that solving the update problem is very simple. Let \verb#State# be a ground PI-List, and let \verb#Update#
be a list of ground fluents. The representation of the next state is then computed in two steps\/:

\begin{itemize}
\item[(1)] First, all prime implicates in \verb#State# that contain either an effect from \verb#Update#, 
  or its negation, are deleted, resulting in \verb#State1#.
\item[(2)] The next state \verb#NextState# is given by the union of \verb#State1# and \verb#Update#.
\end{itemize}

Starting from a ground initial state only ground states are computed.

The correctness of this procedure can be seen e.g.\ as follows\/: In \cite{LiLuMiWo-KR-06,update}
algorithms for computing updates in a Fluent Calculus based upon Description Logics have been developed.
The above update algorithm constitutes a special case of these algorithms.

\subsubsection{The Entailment Problem}

When evaluating a clause $\psi$ against
a ground prime state formula $\phi$,
$\psi$ is first split into the fluent part $\psi_1$, and the non-fluent part $\psi_2$.
It then holds that $\psi$ is entailed by $\phi$
if there is a ground substitution $\theta$ such that
\begin{itemize}
\item $\psi_1\theta$ is subsumed by some prime implicate in $\phi$; or
\item some auxiliary atom $P(\vec x)\theta$ from $\psi_2$ can be derived from its defining Prolog clauses.
\end{itemize}

Computing that the clause $\psi_1$ is subsumed by $\phi$ can be done as follows\/:
\begin{itemize}
\item If $\psi_1$ is a singleton, then it must be a prime implicate of $\phi$ (modulo unification).
\item Otherwise there must be a prime implicate in $\phi$ that contains $\psi_1$ (modulo unification).
\end{itemize}
Hence the entailment problem for ALPprolog can be solved by 
\verb#member#, \verb#memberchk#, and \verb#subset# operations on sorted, duplicate-free lists.

The following example illustrates how reasoning in ALPprolog can be reduced to simple operation on lists.
It also illustrates the limited form of reasoning about disjunctive information available in ALPprolog\/:
\newexample
\begin{example}
Assume that the current state is given by \verb#[[at(gold,4),at(gold,5)]]#.
Then the query \verb#?([at(gold,X)])# fails, 
because we don't consider disjunctive substitutions. 
However, on the same current state the query \verb#?([[at(gold,X),at(gold,Y)]])# succeeds with \verb#X=4# and \verb#Y=5#.
\end{example}

\subsubsection{The Sensing Problem}

Sensing results have to be included into the agent's state knowledge every time a sensing action is performed,
i.e.\ a literal \verb#?(s(X))# is evaluated. This works as follows\/:
\begin{itemize}
\item First we identify the appropriate sensor axiom \verb#sensor_axiom(s(X),Vals)#.
\item Next we identify all the \verb#[X-result_i]-Index-Meaning# triples in \verb#Vals# 
  such that \verb#result_i# matches the observed sensing result, and unify \verb#X# with \verb#result_i#.
\item We then locate the unique \verb#Index-Meaning# s.t.\ the current state entails \verb#Index#.
\item Finally, we adjoin \verb#Meaning# to the current state and transform this union to a PI-list.
\end{itemize}

\subsection{Soundness of ALPprolog}

At the end of section~\ref{ssec:sens}
we have already mentioned that adding sensing results to the action theory at runtime makes the subsequent reasoning
logically unsound wrt.\ the original program plus action theory.
If we add the set of sensing results observed throughout a run of an ALPprolog program, however,
then we can obtain the following soundness result\/:

\newproposition
\begin{proposition}
\label{prop:sound}
Let \Programs be a ALPprolog program on top of an action domain \DomAx. 
Let $\Sigma$ be the union of the sensor results observed during a successful derivation of the ALPprolog
query \Querys with computed answer substitution $\theta$. Then $\DomAx \cup \Program \cup \Sigma \vDash \Query\theta$.
\end{proposition}

\begin{proof}[Proof (Sketch)]
It is well-known that SLD-resolution is sound for any ordinary program atom.
A query $(\exists)\Holds(\phi,s)\!$, where $\phi$ is not a sense fluent, is only evaluated successfully if there is a substitution 
$\theta$ such that $\DomAx \vDash (\forall)\Holds(\phi,s)\theta$. 
Assume we observe the sensing result $R_i\in\Rmc$ for a sense fluent $S(x,s)$. In general we have (cf.\ Definition~\ref{def:sens})\/:
\[
\DomAx \vDash (\forall)\Holds(S(x,s)) \land \bigvee_{R\in\Rmc}x = R, \text{ but } \DomAx \nvDash (\forall)\Holds(S(x,s)) \land x = R_i.
\]
For soundness, we have to add the observed sensing result as an additional assumption to the theory\/: 
$\DomAx\cup \{\Holds(S(R_i,s))\}\vDash (\forall)\Holds(S(x,s))\land x=R_i$.
\end{proof}

\section{Evaluation}
\label{sec:eval}

We have evaluated the performance of ALPprolog via the so-called
Wumpus World \cite{RussellNorvig03} that is a well-known challenge problem in the reasoning about action community. 
Essentially, the Wumpus World is an extended version of the gold-hunter domain from example~\ref{ex:alp}.
The main features that make it a good challenge problem are incomplete information in the form of disjunctions and unknown propositions, 
and reasoning about sensing results.

We have used both Flux and ALPprolog to solve Wumpus Worlds of size up to $32 \times 32\!$.\footnote{The distribution of ALPprolog contains the Wumpus World example for both ALPprolog and Flux.}
We have done this using three different modellings\/:

\begin{itemize} 
\item[(1)] In \cite{D:24} a Flux model is described that uses quantification over variables --- this is beyond ALPprolog.
\item[(2)] We have evaluated both languages on a ground model.
\item[(3)] We have artificially increased the size of the ground model by making the connections between cells part of the state knowledge.
\end{itemize}

A first observation is that both languages roughly scale equally well in all models.
Using (1) Flux is slightly faster than ALPprolog using (2).
Let us then point out that on ground models Flux and ALPprolog maintain the same state representation\/:
Flux also computes the prime implicates.
On the encoding (2) ALPprolog is roughly one order of magnitude faster than Flux,
whereas on (3) the difference is already two orders of magnitude.
The key to the good performance of ALPprolog then is that it handles large state representations well\/:
By encoding states as {\em sorted} lists (of lists) some of the search effort necessary in Flux can be avoided.
If, however, we use Flux' capability of handling quantified variables in the state knowledge for a more concise encoding,
then ALPprolog and Flux are again on par, with Flux even having slightly the edge.
In general, we expect ALPprolog to excel on problem domains that feature large state representations that are not easily compressed using quantification.

It has already been established that Flux gains continuously over standard Golog 
the more actions have to be performed \cite{A:16}. 
As ALPprolog scales as well as Flux
the same holds for ALPprolog and Golog.
The version of Golog with periodically progressed state knowledge is slightly 
slower than Flux \cite{sardin:wumpus}. 

Let us also compare ALPprolog, Flux, and Golog from a knowledge representation perspective\/:
Both ALPprolog and Flux allow the programmer to define new auxiliary predicates for the agent strategy
that are not present in the action theory, a practically very useful feature that is missing from Golog. Also, the propositional
variables used in Golog instead of the finitely many ground terms used in ALPprolog make it hard for the programmer to
fully exploit the power of Prolog's unification mechanism. In this regard Flux, on the other hand, excels in that
the programmer can include fluents containing (possibly quantified) variables in the agent's state knowledge.
Contrary to ALPprolog and Golog, however, Flux does not support arbitrary disjunctions.

\section{Conclusion and Future Work}
\label{sec:conc}

In this work we have presented ALPprolog, an efficient logic programming language for the online control
of autonomous agents in domains that feature incomplete information and sensing. 
On the one hand, it can be argued that the state-of-the-art languages Golog and Flux already successfully address this application domain.
On the other hand, we have shown that ALPprolog excels because of its efficient reasoning 
with large ground state representations, something that we expect to be quite useful in practice.

For future work, there are two interesting directions\/:
On the one hand it would be nice to extend ALPprolog to offline planning.
The disjunctive substitutions in the general ALP proof calculus
provide a powerful form of reasoning about conditional plans, or planning in the presence of sensing
in the sense of \cite{conf/aaai/Levesque96}.

On the other hand we plan to fruitfully apply ALPprolog in the domain of General Game Playing.
General Game Playing \cite{journals/aim/GeneserethLP05} is a new exciting AI research challenge aiming at the 
integration of manifold AI techniques\/: A program (also called a player)
is given an axiomatisation of the rules of a game. 
The player then computes a strategy/heuristic that it uses to play
and hopefully win the game.
The main challenge of General Game Playing consists of constructing suitable heuristics.

However, at its base the player also needs a means to represent, and reason about, the
state of the game.
Up to now the games played in General Game Playing have been restricted to complete 
information \cite{Love2008} --- but clearly games with incomplete information constitute a bigger challenge~\cite{B:60}.
We intend to include techniques from ALPprolog into the successful Flux-based Fluxplayer \cite{schiffel:fluxplayer}.

\paragraph{Acknowledgements.}
We appreciate the helpful comments by the reviewers. This work was partially supported by DFG Grant TH 541/14. C.\ Drescher wishes to acknowledge support by EPSRC Grant EP/G055114/1. M.\ Thielscher is the recipient of an Australian Research Council Future Fellowship (project number~FT\,0991348). He is also affiliated with the University of Western Sydney.

\bibliographystyle{acmtrans}

\end{document}